\documentclass[a4paper,11pt,twoside,notitlepage,reqno]{amsart}

\usepackage{amsmath,amssymb,amsbsy,amsthm}
\usepackage[hmargin=1.4in,vmargin=1.4in]{geometry}
\usepackage{url}
\numberwithin{equation}{section}
\usepackage{tikz}
\usetikzlibrary{matrix}
\usepackage[linktocpage]{hyperref}
\hypersetup{
    colorlinks=true,        
    linkcolor=red,          
    citecolor=red,         
    filecolor=magenta,      
    urlcolor=cyan           
}

\addtocontents{toc}{\protect\setcounter{tocdepth}{1}}

\usepackage{eucal}

\theoremstyle{plain}
	\newtheorem{theorem}{Theorem}
	
		\numberwithin{theorem}{section}
	\newtheorem{lemma}[theorem]{Lemma}
	\newtheorem{proposition}[theorem]{Proposition}
	\newtheorem{corollary}[theorem]{Corollary}
	
\newtheorem{question}[theorem]{Question}
	\newtheorem*{theorem*}{Theorem}
	\newtheorem*{lemma*}{Lemma}
	\newtheorem*{prop*}{Proposition}
	\newtheorem*{cor*}{Corollary}
	\newtheorem*{conj*}{Conjecture}

\theoremstyle{definition}
	\newtheorem{example}[theorem]{Example}
	\newtheorem*{example*}{Example}
	\newtheorem{definition}[theorem]{Definition}
	\newtheorem{remark}[theorem]{Remark}

\begin{document}


\title{Topological invariants for words of linear factor complexity}

\author{Jason P. Bell}
\address{Department of Pure Mathematics\\
University of Waterloo\\
Waterloo, ON N2L 3G1\\
Canada}
\email{jpbell@uwaterloo.ca}
\thanks{The author was supported by NSERC grant RGPIN-2016-03632.}

\begin{abstract} Given a finite alphabet $\Sigma$ and a right-infinite word $w$ over the alphabet $\Sigma$, we construct a topological space ${\rm Rec}(w)$ consisting of all right-infinite recurrent words whose factors are all factors of $w$, where we work up to an equivalence in which two words are equivalent if they have the exact same set of factors (finite contiguous subwords).  We show that ${\rm Rec}(w)$ can be endowed with a natural topology and we show that if $w$ is word of linear factor complexity then ${\rm Rec}(w)$ is a finite topological space.  In addition, we note that there are examples which show that if $f:\mathbb{N}\to \mathbb{N}$ is a function that tends to infinity as $n\to \infty$ then there is a word whose factor complexity function is ${\rm O}(nf(n))$ such that ${\rm Rec}(w)$ is an infinite set.  Finally, we pose a realization problem: which finite topological spaces can arise as ${\rm Rec}(w)$ for a word of linear factor complexity? 
\end{abstract}
\maketitle

\tableofcontents

\section{Introduction}
An important feature in the study of combinatorics of words is the search for meaningful invariants, which give insight into the underlying complexity and structure of given words.  There are numerous examples of such invariants in the theory, such as the critical exponent \cite{BR}, cyclic complexity \cite{CFSZ}, arithmetical complexity \cite{AFF}, abelian complexity \cite{Kar}, Lie complexity \cite{BS1}, Lempel-Ziv complexity \cite{LZ}, letter and word frequencies \cite[Chapt. 1]{AS}, and the factor complexity \cite[Chapt. 10]{AS}.

We recall that, given a right-infinite word $w$ over a finite alphabet $\Sigma$, the \emph{factor complexity function} of $w$, $p_w:\mathbb{N}\to \mathbb{N}$, is the map whose value at $n$ is the number of distinct factors (contiguous finite-length subwords) of $w$ of length $n$.  A result of Morse and Hedlund (see \cite[Theorem 10.2.6]{AS}) shows there is a striking gap: either the factor complexity function of a word $w$ is uniformly bounded, in which case $w$ is eventually periodic; or $p_w(n)\ge n+1$ for all $n$.  

For this reason, it is an especially important project within the field of combinatorics of words to better understand the class of words whose complexity functions lie just on the other side of this gap.  In particular, the Sturmian words are the words $w$ for which $p_w(n)=n+1$ for all $n$ and their study is an important area of focus within the field of combinatorics on words.  More generally, there is the class of words of \emph{linear factor complexity}, which have the property that there is a positive constant $C$ such that $p_w(n)\le Cn$.  This class of words contains many classical examples of words, including all automatic words and all Sturmian words \cite[Chapt. 10]{AS}.  

Much is now known about words of linear factor complexity, and there has been a long history of studying this class of words (see, for just a few examples, \cite[Chapt. 10]{AS} and \cite{Cass, Cassaigne:1996, Mignosi:1989}).  The object of this paper is to introduce new topological invariants, which we hope will provide a coarse taxonomy of words of linear factor complexity, which can then be used to gain greater insight into this important class of words.

We recall that if $w$ is a right-infinite word, then we have a set ${\rm Fac}(w)$, which is the collection of factors of $w$.  Then we have $p_w(n)=\#\{ u\in {\rm Fac}(w)\colon |u|=n\}$, where $|\cdot |$ is the usual length function.

Given two right-infinite words $w,w'$ over $\Sigma$, we declare that $w$ is \emph{equivalent} to $w'$ if $${\rm Fac}(w)={\rm Fac}(w').$$
This induces an equivalence relation $\sim$ on the set of right-infinite words over $\Sigma$, and we let $[w]$ denote the equivalence class of $w$.  
The collection of subsets of ${\rm Fac}(w)$ is a poset under inclusion and this allows us to put a partial order on equivalence classes of right-infinite words over $\Sigma$ by declaring that
\begin{equation} \label{eq:prec} 
[w]\preceq [w']\qquad {\rm if}\qquad {\rm Fac}(w)\supseteq {\rm Fac}(w'),
\end{equation} and we have
 $[w]\preceq [w']$ and $[w']\preceq [w]$ if and only if $[w]=[w']$.  

We recall that a right-infinite word $w$ is \emph{recurrent} if each factor of $w$ occurs infinitely many times in $w$; $w$ is \emph{uniformly recurrent} if for each factor $u$ of $w$, there is some $N=N(u)$ such that every length-$N$ factor of $w$ contains $u$ as a subfactor; $w$ is periodic if $w=uuu\cdots =u^{\omega}$ for some finite word $u$.  These properties are all preserved under this notion of equivalence (see Lemma \ref{lem:equiv}) and in particular, we can speak unambiguously about equivalence classes being recurrent, uniformly recurrent, and periodic.  

The sets that we study in this work are given below.

Let $w$ be a right-infinite word over $\Sigma$.  We define the following sets:
\begin{itemize}
\item the \emph{total spectrum} of $w$ is the set 
\begin{equation}\label{eq:spec1} {\rm Tot}(w)=\{[v]\colon  [w]\preceq [v]\};\end{equation}
\item the \emph{recurrent spectrum} of $w$ is the set
\begin{equation}\label{eq:spec2}{\rm Rec}(w)=\{[v]\colon v~{\rm recurrent}, [w]\preceq [v]\};\end{equation}
\item the \emph{uniformly recurrent spectrum} of $w$ is the set
\begin{equation}\label{eq:spec3}{\rm URec}(w)=\{[v]\colon v~{\rm recurrent}, [w]\preceq [v]\};\end{equation}
\item the \emph{periodic spectrum} of $w$ is the set
\begin{equation}\label{eq:spec4}{\rm Per}(w)=\{[v]\colon v~{\rm periodic}, [w]\preceq [v]\};\end{equation}
\end{itemize}

In particular, we have the containments
\begin{equation}
{\rm Tot}(w)\supseteq {\rm Rec}(w)\supseteq {\rm URec}(w)\supseteq {\rm Per}(w).
\end{equation}

Observe that ${\rm Tot}(w)$, ${\rm Rec}(w)$, ${\rm URec}(w)$, and ${\rm Per}(w)$ are all posets under the ordering $\prec$.
The main result of this paper is the following.
\begin{theorem}\label{thm:main1}
Let $w$ be a right-infinite word over a finite alphabet $\Sigma$ and suppose that there is a constant $C$ such that $p_w(n)\le Cn$ for all $n\ge 1$.  Then $\#{\rm Rec}(w)$ is finite and we have the upper bounds
$$\#{\rm Rec}(w)\le \left( \limsup_{n\to\infty} p_w(n+1)-p_w(n)\right)+\lceil C\rceil!^2 $$ and $$\#{\rm URec}(w)\le  \left( \limsup_{n\to\infty} p_w(n+1)-p_w(n)+1\right) + C^2$$ 
\end{theorem}
Upper bounds for ${\rm Per}(w)$ have been obtained previously \cite{BS1}, where it is shown that $\#{\rm Per}(w)\le \limsup_{n\to\infty} p_w(n+1)-p_w(n)+1$, which a result of Cassaigne \cite{Cass} shows is finite when $w$ has linear factor complexity.  Bounds for ${\rm Per}(w)$ had been previously considered in other contexts \cite{KS}. 

We also note that if $f$ is a weakly increasing function that tends to $\infty$ then there exists a word $w$ with $p_w(n)={\rm O}(nf(n))$ such that ${\rm Per}(w)$ and ${\rm Rec}(w)$ are infinite (see Remark \ref{rem:rec} for details).

We also show that the space ${\rm Rec}(w)$ can be endowed with a natural topology, reminiscent of the Zariski topology in algebraic geometry.  The appeal of having a topology is that one can consider continuous maps and use the additional topological structure to obtain new results for words.  For example, we show the factor complexity function gives a continuous map from ${\rm Rec}(w)$ to the space of $\mathbb{N}$-valued sequences in a natural sense (see Theorem \ref{thm:alex})

The outline of this paper is as follows. In \S2, we prove basic properties of the four spectra defined in Equations (\ref{eq:spec1})--(\ref{eq:spec4}).  In \S3, we define the radical of a word, which we later relate to the uniformly recurrent spectrum.  
In \S4, we develop the basic topological results about ${\rm Rec}(w)$ and prove a result about continuity of the factor complexity function in this framework.  In \S5 we prove Theorem \ref{thm:main1} and we conclude by posing questions about these constructions in \S6.

\section{Basic properties of spectra}
Let $\Sigma=\{x_1,\ldots ,x_d\}$ be a finite set with $d\ge 2$.  Let $w$ be a right-infinite word over $\Sigma$.  We recall that
$${\rm Fac}(w) = \{ v\in \Sigma^*\colon v~{\rm is ~a~factor~of~}w\},$$ which we call the \emph{factor set} of $w$.
If we adjoin an absorbing element ${\bf 0}$, ${\rm Fac}(w)\cup \{{\bf 0}\}$ is a monoid, where we declare that ${\bf 0}\cdot z=z\cdot {\bf 0} ={\bf 0}$ for every $z\in {\rm Fac}(w)$, ${\bf 0}\cdot {\bf 0}={\bf 0}$, and for $z,y\in {\rm Fac}(w)$, and $z\cdot y$ is the concatenation of $z$ and $y$ if it is a factor of $w$ and is ${\bf 0}$ otherwise.  We thus occasionally call ${\rm Fac}(w)$ the \emph{factor monoid} of $w$ when we wish to emphasize the monoidal structure, with the understanding that we are then adjoining a zero element.
We then define the set
\begin{equation}
I(w) = \Sigma^*\setminus {\rm Fac}(w),
\end{equation}
which we call the \emph{ideal} of $w$.  Then $I(w)$ is closed under left and right concatenation by words in $\Sigma^*$. 
In particular, if we let $\mathbb{Q}\{x_1,\ldots ,x_d\}$ denote the free associative $\mathbb{Q}$-algebra on $\Sigma$ then the $\mathbb{Q}$-span of elements of $I(w)$ will be a two-sided ideal, $\langle I(w)\rangle$, of $\mathbb{Q}\{x_1,\ldots ,x_d\}$ and we define the algebra 
\begin{equation}
A_w:=\mathbb{Q}\{x_1,\ldots ,x_d\}/\langle I(w)\rangle,\label{eq:Aw}\end{equation}
 which we call the \emph{monomial algebra} associated to $w$.

The partial order $\prec$ on the equivalence classes of right-infinite words over $\Sigma$ given in Equation (\ref{eq:prec}) can be stated in terms of ideals as follows:
\begin{equation}
\label{eq:prec2}
[w]\preceq [w']\qquad {\rm if}\qquad I(w)\subseteq I(w').\end{equation}

Although equivalent words can be very different from one another in general, many natural combinatorial properties are preserved under this equivalence as the following result shows.

\begin{lemma} Let $\Sigma$ be a finite alphabet and let $w$ and $w'$ be equivalent right-infinite words over $\Sigma$.
Then the following hold:
\begin{enumerate}
\item $w$ is recurrent if and only if $w'$ is recurrent;
\item $w$ is uniformly recurrent if and only if $w'$ is uniformly recurrent;
\item $w$ is periodic if and only if $w'$ is periodic.
\end{enumerate}
\label{lem:equiv}
\end{lemma}
\begin{proof}
Suppose that $w$ is recurrent and that $w'$ is not. Then there is some factor $z$ of $w'$ that does not reoccur in $w'$.  But since $w$ is recurrent and since $z$ is also a factor of $w$, there is some factor $y$ of $w$ such that $zyz$ is a factor of $w$ and hence of $w'$, contradicting that $z$ does not reoccur in $w'$.  Uniform recurrence is proved similarly: if $w$ is uniformly recurrent and $z$ is a factor of $w$ then there is some $N=N(z)$ such that all factors of $w$ of length $N$ contain $z$ as a subfactor; but ${\rm Fac}(w)={\rm Fac}(w')$ and so all factors of $w'$ of length $N$ contain $z$ too, and thus $w'$ is uniformly recurrent.  Finally, if ${\rm Fac}(w)={\rm Fac}(w')$ then $p_w(n)=p_{w'}(n)$ for all $n$ and since a right-infinite word is periodic if and only if it is recurrent and has uniformly bounded factor complexity function, we see that $w$ is periodic if and only if $w'$ is periodic.
\end{proof}
In light of Lemma \ref{lem:equiv}, we say that an equivalence class $[w]$ is recurrent (resp. uniformly recurrent, resp. periodic) if a representative, and hence every representative, is recurrent (resp. uniformly recurrent, resp. periodic).

We now develop the theory of the spectra defined in Equations (\ref{eq:spec1})--(\ref{eq:spec4}).  As noted before, ${\rm Tot}(w)$, ${\rm Rec}(w)$, ${\rm URec}(w)$, and ${\rm Per}(w)$ are posets under $\prec$; furthermore, all elements of ${\rm URec}(w)$ and ${\rm Per}(w)$ are maximal, which is an immediate consequence of the following lemma.
\begin{lemma} Let $\Sigma$ be a finite alphabet and let $w$ be a right-infinite word over $\Sigma$.  Then $[v]\in {\rm Tot}(w)$ is maximal with respect to the order $\prec$ if and only if $[v]$ is uniformly recurrent.  Moreover, for every $[u]\in {\rm Tot}(w)$, there is some $[v]$ in ${\rm URec}(w)$ with $[u]\preceq [v]$.  \label{lem:22}
\end{lemma}
\begin{proof}
Let $[u]\in {\rm Tot}(w)$.  Then by Furstenberg's theorem \cite{Fur} (see also \cite[Exercise 2, p. 337]{AS} and see \cite[Theorem 4.4.9]{B1} for an algebraic proof), there is some uniformly recurrent word $v$ over $\Sigma$ such that ${\rm Fac}(v)\subseteq {\rm Fac}(u)$ and so $[u]\preceq [v]$.  Thus it suffices to prove that elements of ${\rm URec}(w)$ are maximal elements of ${\rm Tot}(w)$.  To this end, let $[v]\in {\rm Urec}(w)$ and let $[v']\in {\rm Urec}(w)$ with 
${\rm Fac}(v')\subseteq {\rm Fac}(v)$. Let $z\in {\rm Fac}(v)$.  Since $v$ is uniformly recurrent, there is some natural number $N$ such that all factors of $v$ of length $N$ have $z$ as a subfactor.  In particular, $z$ is a factor of every length $N$ factor of $v'$ and so $z\in {\rm Fac}(v')$ and so ${\rm Fac}(v)={\rm Fac}(v')$, which shows that $[v]$ is maximal in ${\rm Tot}(w)$.
\end{proof}
\section{The radical of a word}
A quantity that appears to be particularly useful in understanding the structure of infinite words is a subset of the collection of factors of a word that we call the \emph{radical} of a word, due to its connection with the nil radical in ring theory.
\begin{definition}
Let $\Sigma$ be a finite alphabet and let $w$ be a right-infinite word over $\Sigma$.  We define the \emph{radical} of $w$ to be the set of factors $z$ of $w$ with the property that for every finite set $\{y_1,\ldots ,y_n\}$ of factors of $w$, each of which has $z$ as a subfactor, there is some natural number $N=N(y_1,\ldots ,y_n)$ such that $y_{i_1}\cdots y_{i_N}\in I(w)$ for every $(i_1,\ldots ,i_N)\in \{1,\ldots ,n\}^N$.
\end{definition}
It is immediate from the definition that if $z\in {\rm Rad}(w)$ then so is every factor of $w$ that contains $z$.  In particular, the set ${\rm Rad}(w)\cup \{{\bf 0}\}$ is an ideal of the factor monoid.
\begin{example}
Let $\Sigma=\{x,y\}$ and let $w=xyx^2yx^4y x^8yx^{16}y\cdots $.  Then ${\rm Rad}(w)$ is the set of factors of $w$ that contain $y$. 
\end{example}
\begin{proof}  Let $z_1,\ldots ,z_n$ be factors of $w$ that contain $y$ and let $M=\max(|z_1|,\ldots ,|z_n|)$.  Then a concatenation of $4^M$ elements of $\{z_1,\ldots ,z_n\}$ has length at least $4^M$ and length at most $M\cdot 4^M$.  Notice that the last $2M$ letters of this concatenation must have at least two copies of $y$, but this is impossible, $xyx^2y \cdots x^{2^M}y$ is a prefix of $w$ of length $<4^M-2M$ and all subsequent occurrences of $y$ occur at least $2^{M+1}>2M$ positions apart.  It follows that every factor of $w$ containing $y$ is in ${\rm Rad}(w)$.  On the other hand, a factor of $w$ not containing $y$ is of the form $x^i$ and since arbitrarily large powers of $x$ occur in $w$, we obtain the claim.
\end{proof}
The following result gives a useful characterization of the radical of a word.
\begin{proposition} Let $w$ be a right-infinite word over a finite alphabet $\Sigma$.  Then $z$ is in ${\rm Rad}(w)$ if and only if for every natural number $n\ge |z|$ there is some $N=N(n)$ such that all factors of $w$ of length at least $N$ have a subfactor of length $n$ that does not contain $z$ as a factor.
\end{proposition}
\begin{proof}
Suppose that $z\in {\rm Rad}(w)$ and let $n\ge |z|$.  Let $u_1,\ldots ,u_m$ be the factors of $w$ of length $n$ that contain $z$ as a factor.  Then there is some $N$ such that every concatenation of $N$ elements from 
$\{u_1,\ldots ,u_m\}$ is not in ${\rm Fac}(w)$.  In particular, every factor of $w$ of length $Nn$ must have some subfactor of length $n$ that does not contain $z$ as a factor.  Conversely, suppose that $z$ is not in ${\rm Rad}(w)$.  Then there exist factors $a_1,\ldots, a_m$ of $w$ that contain $z$ as a factor such that there are arbitrarily long concatenations of elements from $\{a_1,\ldots ,a_p\}$ that are in ${\rm Fac}(w)$.  Let $T$ denote the collection of words over 
$\{a_1,\ldots ,a_p\}$ that are factors of $w$.  Then $T$ is a factor-closed infinite set and thus by K\"onig's infinity lemma there is some right-infinite word $v$ with ${\rm Fac}(v)\subseteq {\rm Fac}(w)$ and such that every factor of $v$ of length $C:=2\cdot \max(|a_1|,\ldots ,|a_p|)$ contains some $a_i$ (and hence $z$) as a subfactor.  Thus there does not exist a number $N$ such that every factor of $w$ of length at least $N$ has some subfactor of length $C$ that does not contain $z$ as a factor.  The result follows.
\end{proof}

\section{Topology on ${\rm Rec}(w)$}
In this section, we put a natural topology on the set ${\rm Rec}(w)$.  This definition should not be regarded as new, as it is really derived from the Zariski topology on a subspace of the prime spectrum of an associated noncommutative ring.  We nevertheless give a full proof that this proposed topology does indeed fulfill the requirements of being a topological space to illustrate the results. 

Let $w$ be a right-infinite word over a finite alphabet $\Sigma$.  For each subset $S$ of ${\rm Fac}(w)$ that is closed under the process of taking factors, we define a set 
\begin{equation}\mathcal{C}(S)  = \{ [u]\in {\rm Rec}(w) \colon {\rm Fac}(u) \subseteq S\}.\end{equation} 

 
\begin{proposition}  Let $\Sigma$ be a finite alphabet and let $w$ be a right-infinite word over $\Sigma$.
Then we have a topology on ${\rm Rec}(w)$ in which the closed subsets are precisely the sets of the form $\mathcal{C}(S)$ with $S\subseteq {\rm Fac}(w)$ a set closed under the process of taking factors. 
\end{proposition} 
\begin{proof} Notice that ${\rm Rec}(w)=\mathcal{C}({\rm Fac}(w))$ and $\emptyset = \mathcal{C}(\emptyset)$, and so the empty set and ${\rm Rec}(w)$ are closed sets.  Let $\{S_{\alpha}\colon \alpha\in J\}$ be a collection of subsets of ${\rm Fac}(w)$ in which each set is closed under the process of taking factors. Then
$$\bigcap_{\alpha\in J} \mathcal{C}(S_{\alpha}) = \mathcal{C}\left(\bigcap_{\alpha\in J} S_{\alpha}\right),$$ and this shows that our collection of sets is closed under arbitrary intersections.  We now show that our sets are closed under finite unions, which is where we need the use of recurrence in our words.  Let $S_1,\ldots ,S_k$ be a finite collection of subsets of 
${\rm Fac}(w)$ that are closed under the process of taking factors.  Then we claim that 
$\mathcal{C}(\bigcup S_i)=\bigcup_{i=1}^k \mathcal{C}(S_i)$.   To see this, suppose that $[v]\in \bigcup_{i=1}^k \mathcal{C}(S_i)$.  Then there is some $i$ such that every factor of $v$ is in $S_i$ and so $[v]\in \mathcal{C}(\bigcup S_i)$.  Conversely, if $[v]\in \mathcal{C}(\bigcup S_i)$, then every factor of $v$ lies in $S_1\cup \cdots \cup S_k$.
Now suppose that there does not exist an $i$ such that ${\rm Fac}(v)\subseteq S_i$.  Then for each $i\in \{1,\ldots ,k\}$ there is some factor $v_i$ of $v$ such that $v\not \in S_i$.  Since $v$ is recurrent, there exist factors $a_1,a_2,\ldots ,a_k$ such that $v_1a_1 v_2 a_2 \cdots a_{k-1} v_k$ is a factor of $v$.  By definition this means $v_1a_1 v_2 a_2 \cdots a_{k-1} v_k\in S_i$ for some $i$, which is impossible, since $S_i$ is closed under the process of taking factors and $v_i\not\in S_i$.
\end{proof}
\begin{remark} The topology above is closely related to the Zariski topology.  Given a right-infinite word $w$ over a finite alphabet $\Sigma$, one can form the algebra $A_w$ as in Equation (\ref{eq:Aw}).  Then it is a classical result that the collection of prime ideals in a ring can be endowed with the Zariski topology.  An element $[v]\in {\rm Rec}(w)$ naturally corresponds to a prime ideal of $A$, by taking the image of ideal $I(v)$ in $A$.  Then if we identify ${\rm Rec}(w)$ with this distinguished subset of the prime ideals of $A_w$, then the topology on ${\rm Rec}(w)$ is exactly the subspace topology\footnote{Given a topological space $X$, a subset $Y$ inherits a topological structure from $X$ in which the open subsets of $Y$ are precisely the sets of the form $U\cap Y$ with $U$ an open subset of $X$; we call this topology the \emph{subspace topology}.} when we endow the set of prime ideals of $A_w$ with the Zariski topology.
\end{remark}
We recall that a subset $Y$ of a topological space $X$ is \emph{dense} if every closed subset of $X$ that contains $Y$ is necessarily all of $X$; equivalently, every non-empty open subset of $X$ intersects $Y$ non-trivially. 
\begin{remark} If $w$ is a recurrent word, then $[w]$ is a dense point of ${\rm Rec}(w)$; that is, a closed set that contains $[w]$ is necessarily all of ${\rm Rec}(w)$.  Equivalently $[w]$ is in every non-empty open subset of ${\rm Rec}(w)$.  
\label{rem:closedpt}
\end{remark}
\begin{remark} Lemma \ref{lem:22} gives that a point $[v]$ of ${\rm Rec}(w)$ is closed if and only if $[v]$ is uniformly recurrent. 
\end{remark}

Given a factor $z\in {\rm Fac}(w)$ we have a set 
\begin{equation}
\mathcal{U}(z)=\{[v]\in {\rm Rec}(w)\colon z\in {\rm Fac}(v)\}.
\end{equation}  Then
the complement of $U(z)$ is the set $\mathcal{C}(S_z)$, where $S_z$ is the factor-closed subset of ${\rm Fac}(w)$ consisting of factors of $w$ that do not have $z$ as a subfactor.  Thus $\mathcal{U}(z)$ is an open set, and we call sets of this form \emph{principal open subsets} of ${\rm Rec}(w)$.  Notice that every open subset of ${\rm Rec}(w)$ can be written as a union of principal open sets, since if $U$ is an open subset of ${\rm Rec}(w)$ then $U$ is the complement of $\mathcal{C}(S)$ for some factor-closed subset $S$ of ${\rm Fac}(s)$.  If $[v]\in U$, then by definition of $U$ there is some factor $z$ of $v$ that is not in $S$.  Then $\mathcal{U}(z)\subseteq U$ and it contains $[v]$.  Thus $U$ can indeed be expressed as a union of principal open subsets. 

The collection of maps from $\mathbb{N}$ to itself can be identified with $\mathbb{N}^{\mathbb{N}}$ via the correspondence 
$$f:\mathbb{N}\to \mathbb{N} \mapsto \{f(n)\}_{n\in \mathbb{N}}.$$  There are several natural topologies on $\mathbb{N}$, but since we wish to emphasize its structure as an ordered set, we use the poset topology (also called the Alexandrov topology), in which the non-empty open subsets of $\mathbb{N}$ are precisely the right-infinite rays $U_a:=\{a,a+1,a+2,\ldots\}$ for $a\ge 0$.  Then it is not difficult to check that $\mathbb{N}$ becomes a topological space with these sets, along with the empty set, as the open sets.  It is then natural to give $\mathbb{N}^{\mathbb{N}}$ the product topology, which is the unique topology that makes $\mathbb{N}^{\mathbb{N}}$ a product of copies of $\mathbb{N}$ in the category of topological spaces.  The open subsets of $\mathbb{N}^{\mathbb{N}}$ are then those sets that can be written as unions of sets $U_{a_1,\ldots ,a_s;b_1,\ldots ,b_s}$, where $s\ge 1$ and $a_1,\ldots ,a_s,b_1,\ldots ,b_s\in \mathbb{N}$ and
$U_{a_1,\ldots ,a_s;b_1,\ldots ,b_s}$ is the set of maps $f:\mathbb{N}\to \mathbb{N}$ such that $f(a_i)\ge b_i$ for $i=1,\ldots ,s$.  We call this the \emph{Alexandrov product} topology on $\mathbb{N}^{\mathbb{N}}$.  

We recall that a function $f$ from a topological space $X$ to another topological space $Y$ is \emph{continuous} if $f^{-1}(U)$ is an open subset of $X$ for every open subset $U$ of $Y$.  
\begin{theorem}
The factor complexity map $p:{\rm Rec}(w)\to \mathbb{N}^{\mathbb{N}}$, which sends 
$[v]\in {\rm Rec}(w)$ to the factor complexity function $p_v: {\mathbb{N}}\to \mathbb{N}$, is an order-reversing continuous map when $\mathbb{N}^{\mathbb{N}}$ is endowed with the Alexandrov product topology.
\label{thm:alex}
\end{theorem}
\begin{proof}
If $[v]\preceq [v']$ then ${\rm Fac}(v')\subseteq {\rm Fac}(v)$ and so $p_{v'}(n)\le p_v(n)$ for every $n$, which gives that $p$ is order-reversing. 
It suffices to show that the preimage of an open set of $\mathbb{N}^{\mathbb{N}}$ under $p$ is open.  Since open sets are unions of principal open sets, it considers to prove this for such sets.

 Consider 
$$V:=p^{-1}(U_{a_1,\ldots ,a_s;b_1,\ldots ,b_s}) = \{[v]\in {\rm Rec}(w)\colon p_v(a_i) \ge b_i~{\rm for~}i=1,\ldots ,s\}.$$
Let $S_i$ denote the set of distinct factors of $w$ of length $a_i$.  Then let $\mathcal{T}$ denote the finite collection of sets
$(T_1, T_2,\ldots ,T_s)$, where $T_i\subseteq S_i$ and $|T_i|=b_i$.  Then $[v]\in V$ if and only if there is some 
$(T_1,\ldots, T_s)\in \mathcal{T}$ such that each $T_i$ is contained in the set of distinct factors of $v$ of length $a_i$.  Thus 
$[v]$ is in the open set $$\mathcal{U}(T_1,\ldots ,T_s):=\bigcap_{z\in \bigcup_{i=1}^s T_i} \mathcal{U}(z).$$ 
Hence $V$ is the union of the open sets $\mathcal{U}(T_1,\ldots ,T_s)$ as we range over $(T_1,\ldots ,T_s)\in \mathcal{T}$, which is open.  
\end{proof} 
\begin{corollary} Let $\Sigma$ be a finite alphabet and let $w$ be a right-infinite word over $\Sigma$.  If $h:\mathbb{N}\to \mathbb{N}$ is a map then the set of $[v]\in {\rm Rec}(w)$ such that $p_v(n)< h(n)$ for every $n$ is a closed subset of ${\rm Rec}(w)$.
\end{corollary}
\begin{proof}
For each $i$, the set $\mathcal{G}_i$ of maps $g:\mathbb{N}\to \mathbb{N}$ with $g(i)<h(i)$ is a closed set of $\mathbb{N}^{\mathbb{N}}$, as it is the complement of the open set $U_{i;h(i)}$.  The preimage of $\mathcal{G}_i$ under the map $p$ given in the statement of Theorem \ref{thm:alex} is closed by Theorem \ref{thm:alex} and thus $\{[v]\in {\rm Rec}(w)\colon p_v(i)<h(i)\}$ is closed.  Since an intersection of closed sets is again closed, we obtain the result.
\end{proof}

The next result shows that when the radical of a word is trivial, the uniformly recurrent spectrum is dense in the recurrent spectrum.  This is a translation of a well known results about the structure of rings (in particular, results about radicals).  
\begin{proposition} Let $\Sigma$ be a finite alphabet and let $w$ be a right-infinite word over $\Sigma$.  Then ${\rm URec}(w)$ is a dense subset of $\mathcal{C}(S)$, where $S={\rm Fac}(w)\setminus {\rm Rad}(w)$.  In particular, if ${\rm Rad}(w)$ is empty, then ${\rm URec}(w)$ is dense in ${\rm Rec}(w)$.  
\label{prop:URecdense}
\end{proposition}
\begin{proof} Recall ${\rm Rad}(w)\cup \{{\bf 0}\}$ is an ideal of the factor monoid and so $S:={\rm Fac}(w)\setminus {\rm Rad}(w)$ is a factor-closed set.  Thus $\mathcal{C}(S)$ is a closed subset of ${\rm Rec}(w)$.  We first claim that if $v$ is uniformly recurrent then $[v]\in \mathcal{C}(S)$.  
To see this, observe that if $v\not\in \mathcal{C}(S)$ then $v$ has a factor $z$ that is in ${\rm  Rad}(w)$.  Since $v$ is uniformly recurrent, there is some $N$ such that every factor of $v$ of length $N$ contains $z$ as a factor.  Let $u_1,\ldots ,u_m$ denote the factors of $v$ of length $N$.  Then by definition of the radical, there is some $M$ such that $u_{i_1}\cdots u_{i_M}\in I(w)\subseteq I(v)$ for every $i_1,\ldots ,i_M\in \{1,\ldots ,m\}$.  But this is a contradiction, since $v$ is a right-infinite word over the alphabet $\{u_1,\ldots ,u_m\}$.  Thus ${\rm URec}(w)\subseteq \mathcal{C}(S)$. 

Now suppose that ${\rm URec}(w)$ is not dense in $\mathcal{C}(S)$.  Then there is an open set $U$ of ${\rm Rec}(w)$ that intersects $\mathcal{C}(S)$ non-trivially and which intersects ${\rm URec}(w)$ trivially.   Since open sets of ${\rm Rec}(w)$ are unions of principal open sets, we may assume without loss of generality that there is some factor $z$ of $w$ such that $U=\mathcal{U}(z)$.  
Now let $\mathcal{T}$ denote the set of factors of $w$ that contain $z$ as a subfactor and are not in ${\rm Rad}(w)$.  Then $\mathcal{T}$ is infinite since $z$ is not in the radical of $w$.  Now since $z$ is not in the radical of $w$, there are factors $y_1,\ldots ,y_d$ of $w$, each of which contains $z$ as a subfactor, such that there are arbitrarily long concatenations of $y_1,\ldots ,y_d$ that lie in ${\rm Fac}(w)$.  We choose such $y_1,\ldots ,y_d$ with $d$ minimal.  

If $d=1$ then by assumption $y_1^n$ is a factor of $w$ for every $n$ and hence $[y_1^{\omega}]\in U$, which contradicts the fact that $U$ contains no uniformly recurrent classes.  Thus we may assume $d\ge 2$.  Now we claim that $y_1,\ldots ,y_d$ are in $\mathcal{T}$.  To see this, suppose that some $y_i$ is not in $\mathcal{T}$. After reindexing, we may assume that $i=d$, and so $y_d\in {\rm Rad}(w)$.  Then by minimality of $d$, there is some $N$ such that all concatenations of $\{y_1,\ldots ,y_{d-1}\}$ of length $N$ lie outside of ${\rm Fac}(w)$.  Thus every concatenation of length $\ge N$ of elements of $\{y_1,\ldots ,y_d\}$ that lies in ${\rm Fac}(w)$ is a concatenation of elements from the finite set $\{uy_du'\}$ where $u$ and $u'$ range over (possibly empty) concatenations of $\{y_1,\ldots ,y_{d-1}\}$ of length at most $N-1$. Then since each $uy_du'$ is in the radical, we see there is some $M$ such that all $M$-fold concatenations of elements from $\{uy_du'\}$ lie outside of ${\rm Fac}(w)$, which contradicts our assumption that there are arbitrarily long concatenations of $y_1,\ldots ,y_d$ that lie in ${\rm Fac}(w)$.

Now we let $\mathcal{X}$ denote the set of factors of concatenations of elements from $\{y_1,\ldots ,y_d\}$ that lie in ${\rm Fac}(w)$.  Then $\mathcal{X}$ is factor-closed and infinite and so by K\"onig's infinity lemma there is some right-infinite word $v=y_{i_1}y_{i_2}\cdots $ such that ${\rm Fac}(v)\subseteq {\rm Fac}(w)$.  Moreover, since each $y_i$ contains $z$ as a factor, there is some $N$ such that every factor of $v$ of length $N$ contains $z$ as a subfactor.  By Furstenberg's theorem \cite{Fur}, there is some uniformly recurrent word $v'$ with ${\rm Fac}(v')\subseteq {\rm Fac}(v)$ and since $z$ is uniformly recurrent in $v$, $z$ is necessarily a factor of $v'$ and so $v'\in \mathcal{U}(z)$, a contradiction.  
The result follows.
 \end{proof}

\section{Words of linear factor complexity}
Given a right-infinite word $w$ over $\Sigma$, we let 
$p_w(n)$ denote the number of elements of ${\rm Fac}(w)$ of length $n$.  Then either $p_w(n)={\rm O}(1)$ or $p_w(n)\ge n+1$ for every $n$.  We now consider words $w$ with the property that there is some $C>1$ such that $n+1\le p_w(n)\le Cn$ for $n\ge 1$.  
Then for such words we have the following result, which shows that in some sense the recurrent spectrum of $w$ is well behaved.
\begin{theorem} Let $\Sigma$ be a finite word and let $w$ be a right-infinite word over $\Sigma$ of linear factor complexity.  Then the following hold:
\begin{enumerate}
\item lengths of chains in ${\rm Rec}(w)$ are uniformly bounded; 
\item if $w$ is recurrent then the union of factor sets of strictly larger recurrent words is proper; i.e.,
$$\bigcup_{\{[v]\in {\rm Rec}(w), [v]\neq [w]\}} {\rm Fac}(v) \subsetneq {\rm Fac}(w);$$
\item if $w$ is recurrent then $\#{\rm Per}(w)\le \limsup_{n\to\infty} p_w(n)/n$.
\end{enumerate}
\label{thm:trich}
\end{theorem}
\begin{remark} A subset $V$ of a topological space $X$ is \emph{locally closed} if $V$ is the intersection of an open subset and a closed subset of $X$.  Then Theorem \ref{thm:trich} says that points in ${\rm Rec}(w)$ are locally closed.  To see this, observe that if $v\in {\rm Rec}(w)$, then $[v]$ is either periodic, in which case it is closed (and hence locally closed), or it has linear factor complexity, in which case
$$Y:=\bigcup_{\{[u]\in {\rm Rec}(v), [u]\neq [v]\}} {\rm Fac}(u) \subsetneq {\rm Fac}(v).$$
Then since every recurrent word whose factors are strictly contained in ${\rm Fac}(v)$ has factors contained in $Y$, we have $$\{[v]\} = \mathcal{C}(Y)^{c}\cap \mathcal{C}({\rm Fac}(v)).$$  
\end{remark}
In fact, finiteness of ${\rm Per}(w)$ holds for every word of linear factor complexity \cite{BS1}, without requiring the assumption that $w$ be recurrent, although the bound will not, in general, be as good as the one we obtain in the recurrent case.  

We begin with a lemma.  We let $\prec$ denote strict inequality and $\preceq$ to denote non-strict inequality of equivalence classes.  Given a factor $u$ of a right-infinite word $w$, we let $p_w(n;u)$ denote the number of factors of $w$ of length $n$ that contain $u$. We use an estimate that is a translation of estimates from \cite{BSmok}.

\begin{lemma} \label{lem:est} 
Let $w$ be a right-infinite aperiodic recurrent word with linear factor complexity and let $u$ be a factor of $w$.  Then $p_w(n;u)\ge n+1-|u|$ for all sufficiently large $n$.  
\end{lemma}
\begin{proof}
For $p,q\ge 0$ we let
$W_{p,q}$ denote the set of factors of $w$ of the form $aub$ with $|a|=p$, $|b|=q$.
Suppose first that for each $(p,q)$ we have
$W_{p,q}$ is not contained in the union
$$\bigcup_{\{(p',q')\colon p'+q'=p+q, p'<p\}} W_{p',q'}.$$ Then for each $(p,q)\in \mathbb{N}^2$ we can pick a word $w_{p,q}$ in $W_{p,q}$ that is not in
$$\bigcup_{\{(p',q')\colon p'+q'=p+q, p'<p\}} W_{p',q'},$$ and so by construction the words $w_{p,q}$ are distinct and we have at least $n+1$ distinct factors of $w$ that contain $u$ of length $n+|u|$, which gives that 
$p_w(n;u) \ge n+1-|u|$ for $n\ge |u|$.  

Hence we may assume that there exists $(p,q)\in \mathbb{N}^2$ such that 
$W_{p,q}$ is contained in the union of $W_{p',q'}$ as $(p',q')$ ranges over the set with $p'<p$ and $p'+q'=p+q$. We may assume without loss of generality that $p\ge q$.  We now claim that every factor of $w$ that contains $u$ can be expressed in the form $aub$ with $|a|< p$ or $|b|<q$. To see this, suppose that this is not the case and pick a factor $y=aub$ of $w$ that contains $u$ that is not of this form with $|a|$ minimal.  
Then $|a|\ge p$ and $|b|\ge q$. So we may write $a=a_1a_0$ and $b=b_0b_1$ with $|a_0|=p$ and $|b_0|=q$.  Since $a_0 u b_0\in W_{p,q}$, there exist $a_0'$ and $b_0'$ with $|a_0'|=p'<p$ such that $a_0ub_0 = a_0' u b_0'$.
So $aub = a_1 a_0' u b_0'b_1$.  But $|a_1a_0'|<|a|$, which contradicts the minimality of $|a|$ and so we obtain the claim.

We now show that in this case that all factors of $w$ of length at least $\max(p,q)+|u|$ must have $u$ as a factor.  To see this, let $y$ be a factor of $w$ of length $\ge \max(p,q)+|u|$. Since $w$ is recurrent, there exist words $a$ and $b$ such that $yauby$ is a factor of $w$.  But every factor of $w$ that contains $u$ can be written in the form $cud$ with $|c|<p$ or $|d|<q$ and since $y$ is both a prefix and suffix of $yauby$ and it has length at least $\max(p,q)+|u|$, we see $y$ must contain $u$ as a factor.  Thus $p_w(n;u) = p_w(n) \ge n+1$ whenever $n\ge \max(p,q)+|u|$.  The result follows.
\end{proof}
\begin{corollary} Let $w$ be a right-infinite recurrent word with linear factor complexity. If $[w]\prec [v]$ then
$$\limsup_{n\to\infty} p_v(n)/n \le  \left(\limsup_{n\to\infty} p_w(n)/n \right)-1.$$
\label{cor:a}
\end{corollary}
\begin{proof} Since $I(v)$ strictly contains $I(w)$, there is some $u\in \Sigma^*$ that is a factor of $w$ but not of $v$. Let $C$ denote the length of $u$. Then by Lemma \ref{lem:est}, for sufficiently large $n$, there are at least $n-C+1$ words of length $n$ in ${\rm Fac}(w)$ that contain $u$ as a factor.  Since none of these words can be factors of $v$, we have
$$p_w(n)\ge p_v(n)+n-C+1$$ for $n$ sufficiently large, and so for every $\epsilon>0$ we have
$p_w(n)/n \ge p_v(n)/n + 1-\epsilon$ for $n$ sufficiently large, which gives the result.
\end{proof}
We recall that a poset $(\mathcal{P},\le)$ satisfies the \emph{ascending chain condition} if whenever 
$$p_1\le p_2\le \cdots $$ is a chain of elements of $\mathcal{P}$, there is some $n$, which depends upon the chain, such that $p_n=p_{n+1}=\cdots $.  The descending chain condition is defined analogously.  The following result immediately gives Theorem \ref{thm:trich} (1).
\begin{corollary} \label{cor:acc} Let $w$ be a right-infinite word of linear factor complexity.  Then all chains in ${\rm Rec}(w)$ have length at most $1+\limsup_{n\to\infty} p_w(n)/n$.  In particular, ${\rm Rec}(w)$ satisfies both the ascending and descending chain conditions.
\end{corollary} 
\begin{proof}
Suppose that $[v_0]\prec [v_1]\prec \cdots \prec [v_m]$ is a strictly ascending chain in ${\rm Rec}(w)$.  Let $C=\limsup_{n\to\infty} p_w(n)/n$.  Then we must show that $m\le C$.
To see this, since ${\rm Fac}(v_0)\subseteq {\rm Fac}(w)$, we have $\limsup_{n\to\infty} p_{v_0}(n)/n\le C$.  Then since $v_1\in {\rm Rec}(v_0)$ and $[v_1]\succ [v_0]$,
$\limsup_{n\to\infty} p_{v_1}(n)/n\le C-1$ by Corollary \ref{cor:a}.  An induction argument using Corollary \ref{cor:a} then gives that $\limsup_{n\to\infty} p_{v_i}(n)/n\le C-i$ for $i=1,\ldots ,m$ and so
$0\le \limsup_{n\to \infty} p_{v_m}(n)/n \le C-m$, which gives the result.
\end{proof}
\begin{proof}[Proof of Theorem \ref{thm:trich} (2)] Let $C=\limsup_{n\to\infty} p_w(n)/n$.  For a factor $u$ of $w$, we let $p_w(n;u)$ denote the factors of $w$ of length $n$ that contain $u$.  Then since $w$ is recurrent, 
by Lemma \ref{lem:est} we have $p_w(n;u)\ge n+1-|u|$.
In particular, $\alpha_u:=\liminf (p_w(n;u)/n \in [1,C)$ for every factor $u$ of $w$.
Let $\alpha$ denote the infimum of all $\alpha_u$ with $u$ a factor of $w$.

Then we claim that if $\alpha_u<\alpha+1$ then $u$ cannot be a factor of $v$ when $[v]\in {\rm Rec}(w)$ is not periodic and $[v]\neq [w]$.  To see this, suppose that $[v]\in {\rm Rec}(w)$ is not periodic and $[v]\neq [w]$ and suppose that $u$ is a factor of $v$.  Pick $\epsilon \in (0,1+\alpha-\alpha_u)$.  Then from Lemma \ref{lem:est}
$p_v(n;u)\ge n+1-|u|$.  Pick a factor $y$ of $w$ that is not a factor of $v$.  Then since $w$ is recurrent, there is a factor of $w$ of the form $yau$ and this is not a factor of $v$.  Now consider the number of factors of $w$ of length $n$ that contain $u$,
$p_w(n;u)$. Since a factor of $w$ that contains $yau$ is never a factor of $v$, we have
 $$p_w(n,u) \ge p_v(n;u) + p_w(n,yau).$$ By definition of $\alpha$ we have $p_w(n,yau) \ge (\alpha-\epsilon)n$ for all sufficiently large $n$ and so
 $$p_w(n,u) \ge n+1-|u| + (\alpha-\epsilon)n,$$ which gives $\alpha_u \ge 1+\alpha-\epsilon > \alpha_u$, a contradiction.  The result follows.  
Thus we obtain the claim.  

Now we pick a factor $u'$ of $w$ with $\alpha_{u'}<\alpha+1$.  By \cite{BS1} there are only finitely many $[v]\in {\rm Per}(w)$ and since $w$ is recurrent we can then pick a factor $u''$ of $w$ that contains $u'$ as a factor and is not a factor of any of the elements of ${\rm Per}(w)$.  Then by construction, $\alpha_{u''} \le \alpha_{u'} < \alpha+1$ and thus $u''$ is not a factor of $v$ for $[v]\in {\rm Rec}(w)$ with $[v]\neq [w]$ and so we obtain the desired result.

\end{proof}
\begin{proof}[Proof of Theorem \ref{thm:trich} (3)]
Pick a positive integer $C$ such that $p_w(n) < Cn$ for $n$ sufficiently large.  We claim that $\# {\rm Per}(w)\le C$.  To see this, suppose that there are pairwise distinct $[u_1^{\omega}],\ldots ,[u_{C+1}^{\omega}]$ in ${\rm Per}(w)$.  Then we can pick a positive integer $D$ such that no cyclic permutation of 
$u_i^D$ is a factor of $u_j^{\omega}$ for $i\neq j$.  We may also assume that no $u_i$ is a power of a strictly shorter word.  Let $L$ be the maximum of the lengths of $u_1^D,\ldots ,u_{C+1}^D$.  Since $w$ has linear factor complexity and each of $u_1^{\omega},\ldots ,u_{C+1}^{\omega}$ have factor complexities that are ${\rm O}(1)$, there exists a factor $y$ of $w$ of length $>2L$ such that none of $u_1^D,\ldots ,u_{C+1}^D$ are factors of $y$.  Then since $w$ is recurrent, for each $n\ge D$ and each $i\in \{1,2,\ldots ,C+1\}$, we can find a factor $u_i^n ayb$ of $w$ with $|b|\ge n$.  In particular, since $u_i^D$ is not a factor of $y$, we see that $u_i^n ayb$ is necessarily of the form $u_i^{n'} b'$ with $n'\ge n$ and $b'$ not having $u_i^D$ as a prefix and $|b'|\ge n$.  In particular, $u_i^{n'} b'$ has a factor of the form 
$$z_{n,i}:=u_i^{\lfloor n/|u_i|\rfloor} c$$ with $|c|\ge n$ and $u_i$ not a prefix of $c$.  Then for each $1\le p<q\le |z_{n,i}|$, we let $z_{n,i}(p,q)$ denote the factor of $z_{n,i}(p,q)$ consisting of the word beginning at the $i$-th position of $z_n$ and ending at the $j$-th position.  Then we claim that
$$\{z_{n,i}(p,p+n)\colon i=1,\ldots ,C+1, 1\le p\le n-3L\}$$ are pairwise distinct.  By construction, these are factors of $w$ of length $n$ and so 
$p_w(n) \ge (C+1)(n-3L) > Cn$ for $n$ sufficiently large, so once we have established the claim, we get a contradiction and obtain the result.

Notice that the first $n-|u_i|$ letters of $z_{n,i}$ is a prefix of $u_i^{\omega}$ and so if $p\le n-3L\le n-2L-|u_i|$, then 
$z_{n,i}(p,p+n)$ has a prefix of length $\le 2L$ that contains $u_i^D$ and does not contain $u_j^D$ for $j\neq i$.  In particular, if $i\neq j$ then
$z_{n,i}(p,p+n)\neq z_{n,j}(p',p'+n)$ for $p,p'\le n-3L$.  Thus it suffices to show that for a fixed $i$ the words $z_{n,i}(p,p+n)$ with $1\le p \le n-3L$ are pairwise distinct.  So suppose that $z_{n,i}(p,p+n)=z_{n,i}(p',p'+n)$ for some $p,p'$ with $1\le p<p'\le n-3L$.  Since some cyclic permutation of $u_i$ is  prefix of $z_{n,i}(p,p+n)$ for $p\le n-3L$, and since the cyclic permutation depends upon $p~(\bmod~|u_i|)$ and since all cyclic permutations of $u_i$ are distinct, we see that $p'\equiv p~(\bmod ~|u_i|)$.  Then there is some $r\in \{0,\ldots ,|u_i|-1\}$ such that $p=|u_i| m_0 + 1+r$ and $p'=|u_i| m_1 +r$ with $0\le m_0<m_1$.  Let $a$ denote the last $|u_i|-r$ letters of $u_i$.  Then 
by construction $z_{n,i}(p,p+n) = a u_i^{\lfloor n/|u_i|\rfloor - m_0-1}c_0$ for some prefix $c_0$ of $c$ and
$z_{n,i}(p',p'+n) = a u_i^{\lfloor n/|u_i|\rfloor - m_1-1}c_1$ for some prefix $c_1$ or $c$.  Then since $m_1>m_0$ we see that $u_i$ must be a prefix of $c_1$, a contradiction.  The result follows.
\end{proof}
\begin{lemma} \label{lem:OK1}
Let $s$ be a positive integer and suppose that $v_1,\ldots ,v_s$ are right-infinite recurrent words over a common finite alphabet such that $[v_i]\npreceq [v_j]$ for $i\neq j$. Then there are factors $y_1,\ldots ,y_s$ of $v_1,\ldots ,v_s$ respectively such that $y_i$ is not a factor of $v_j$ for $i\neq j$.
 \end{lemma}
\begin{proof}
By hypothesis, for each $i\neq j$ there is some factor $z_{i,j}$ of $v_i$ that is not a factor of $v_j$.  Then there is necessarily a finite prefix $y_i$ of $v_i$ that contains each $z_{i,j}$ as a subfactor for $j\neq i$.
Then by construction $y_i$ is a factor of $v_i$ but is not a factor of $y_j$ for $j\neq i$.
\end{proof}

\begin{lemma} \label{lem:OK}
Let $w$ be a word with $p_w(n)\le Cn$. Then the poset ${\rm Rec}(w)$ has at most $C$ non-periodic minimal elements.  \end{lemma}
\begin{proof}
Suppose that there are pairwise distinct non-periodic minimal elements $$[v_1],\ldots, [v_s]$$ with $s>C$ of ${\rm Rec}(w)$.  Then since $[v_i]\not\preceq [v_j]$ for $i\neq j$, by Lemma \ref{lem:OK1} there are factors $y_1,\ldots, y_s$ of $v_1,\ldots ,v_s$, respectively, such that $y_i$ is not a factor of $v_j$ for $i\neq j$. Now since each $v_i$ is aperiodic and recurrent, $p_{v_i}(n;y_i)\ge n+1-|y_i|$ for all $n$, by Lemma \ref{lem:est}.  Moreover, if $S_i$ denotes the set of factors of $v_i$ that contain $y_i$ as a factor, then by construction $S_1,\ldots ,S_s$ are pairwise disjoint and so 
$$p_{w}(n)\ge \sum_{i=1}^s p_{v_i}(n;y_i) > Cn$$ for $n$ sufficiently large, a contradiction.  
\end{proof}
\begin{theorem} Let $w$ be a right-infinite word and suppose that there is a positive constant $C$ such that $p_w(n)\le Cn$ for all $n\ge 1$.  Then $\#{\rm Rec}(w)$ is at most $\#{\rm Per}(w)+\lceil C \rceil !^2-1$.\label{thm:ceil}
\end{theorem}
\begin{proof}
We define ${\rm Rec}_{\rm ap}(u) = {\rm Rec}(u)\setminus {\rm Per}(u)$ for a right-infinite word $u$.
Notice that the result is vacuously true for $C\in (0,1)$.  Suppose that this is not the case for larger $C$.  Then let $C_0$ denote the infimum of all real numbers $C$ of the form $\limsup p_w(n)/n$ as $w$ ranges over right-infinite words of linear factor complexity for which $\#{\rm Rec}(w)$ is larger than the stated bound.  Then we can pick a right-infinite word $w$ with $\limsup p_w(n)/n := C < C_0+1/2$, and with $C>1$.  

By Corollary \ref{cor:acc}, ${\rm Rec}(w)$ satisfies the descending chain condition and therefore every $[v]\in {\rm Rec}(w)$ has the property that there is some minimal element $v'$ of the poset such that $[v]\succeq [v']$.  By Lemma \ref{lem:OK}, ${\rm Rec}(w)$ has at most $C$ non-periodic minimal elements.  Thus we see that ${\rm Rec}_{\rm ap}(w)$ has exactly $t$ minimal elements for some integer $t\le C$, and we let $v_1,\ldots ,v_t$ denote these elements.
Now applying Lemma \ref{lem:OK} again, the poset ${\rm Rec}_{\rm ap}(v_i)\setminus \{[v_i]\}$ contains at most $C$ minimal elements for $i=1,\ldots ,t$ and by Corollary \ref{cor:a} if $u_i$ is an aperiodic minimal element of the poset ${\rm Rec}(v_i)\setminus \{[v_i]\}$, then 
$\limsup p_{u_i}(n)/n \le C-1 <C_0-1/2$.  Then by our choice of $C_0$, we see that 
${\rm Rec}_{\rm ap}(u_i) \le \lceil C-1 \rceil !^2-1$ and since ${\rm Rec}_{\rm ap}(v_i) =[v_i]\cup \bigcup {\rm Rec}_{\rm ap}(u)$, where $u$ ranges over all minimal elements of ${\rm Rec}(v_i)\setminus \{[v_i]\}$.  In particular, we see that
$$|{\rm Rec}_{\rm ap}(v_i)|\le 1+ C (\lceil C-1 \rceil !^2-1) < C\lceil C-1 \rceil !^2,$$ since $C>1$ and $C-1<C_0-1/2$.  Since ${\rm Rec}_{\rm ap}(w)=\bigcup_{i=1}^t {\rm Rec}_{\rm ap}(v_i)$ and since $t\le C$, we see that
$$|{\rm Rec}_{\rm ap}(w)|<C^2\lceil C-1 \rceil !^2\le \lceil C \rceil !^2.$$ Since $|{\rm Rec}_{\rm ap}(w)|$ is an integer, we obtain the desired inequality.
\end{proof}
We may now prove our main result. 
\begin{proof}[Proof of Theorem \ref{thm:main1}]  The inequality for $\#{\rm Rec}(w)$ follows from Theorem \ref{thm:ceil} along with the inequality $\#{\rm Per}(w)\le \limsup_{n\to\infty} (p_w(n+1)-p_w(n)+1)$ from \cite{BS1}.  For the inequality for $\#{\rm URec}(w)$, observe that it again suffices to show that $\#{\rm URec}(w)-\#{\rm Per}(w)\le C^2$.  Lemma \ref{lem:OK} shows that the collection of minimal elements of the poset ${\rm Rec}(w)\setminus {\rm Per}(w)$ is at most $C$.  We let $[v_1],\ldots ,[v_s]$ denote these minimal elements, where $s\le C$.  We claim that there are at most $C$ aperiodic uniformly recurrent elements in ${\rm URec}(v_i)$.  Once we have established this claim, we will then have
$\#({\rm URec}(w)\setminus {\rm Per}(w))\le s\cdot C\le C^2$, as desired.  To establish the claim, suppose that for some $i$, ${\rm URec}(v_i)$ has aperiodic uniformly recurrent pairwise distinct elements $[z_1],\ldots ,[z_t]$ with $t>C$.   Then by Lemma \ref{lem:22}, $[z_1],\ldots ,[z_t]$ are maximal elements of ${\rm Rec}(v_i)$ and hence $[z_i]\npreceq [z_j]$ for $i\neq j$.  Thus by Lemma \ref{lem:OK1}, for each $j$ there is a factor $y_j$ of $z_j$ that is not a factor of $z_k$ for $k\neq j$.  By the uniformly recurrent property, there is some $N$ such that every factor of $z_j$ of length $N$ contains $y_j$ for $j=1,\ldots ,t$.  It follows that for $n\ge N$ the factors of $z_1,\ldots ,z_t$ of length $n$ are disjoint sets and so for $n\ge N$, 
$$Cn \ge p_{v_i}(n) \ge \sum_{j=1}^t p_{z_j}(n) \ge t\cdot (n+1),$$ and so we get a contradiction.  The result now follows.
\end{proof}
\begin{remark}
The construction given in \cite[\S5]{BS1} shows that for every weakly increasing function $f:\mathbb{N}\to \mathbb{N}$, there is a word $w=w(f)$, with $p_w(n)={\rm O}(nf(n))$ such that $\#{\rm Per}(w)=\infty$.  In particular, $\#{\rm Rec}(w)$ is infinite for these words.
\label{rem:rec}
\end{remark}
\section{Concluding remarks}
Theorem \ref{thm:main1} shows that ${\rm Rec}(w)$ is a finite poset whenever $w$ is a word of linear factor complexity.  An interesting question is to characterize the posets that can arise.
\begin{question} Can one give a characterization of the posets that can be realized as ${\rm Rec}(w)$ where $w$ is a word of linear factor complexity?
\end{question}
The software Walnut \cite{Mousavi:2016} has had a revolutionary impact on the theory of automatic words.  Applying the software in practice requires expressing potential theorem statements in a first-order language (see \cite{Charlier&Rampersad&Shallit:2012}).  It would be interesting to know whether there are related algorithms for determining $\#{\rm Rec}(w)$, $\#{\rm URec}(w)$, and $\#{\rm Per}(w)$ for automatic words.  
\begin{question} Can one give a decision procedure that takes automatic words $w$ as input and outputs the sizes of $\#{\rm Rec}(w)$, $\#{\rm URec}(w)$, and $\#{\rm Per}(w)$ as output?
\end{question}
\section*{Acknowledgments}
I thank Jeffrey Shallit for many helpful comments and suggestions.  I also thank the anonymous referee for many useful comments and corrections. 

\end{document}